\documentclass[11pt,a4paper]{amsart}

\usepackage[latin1]{inputenc}
\usepackage[english]{babel}
\usepackage{amsmath}

\usepackage{amssymb, mathabx}
\usepackage[numbers]{natbib}
\usepackage{graphicx}
\usepackage{braket}
\usepackage[pdftex,plainpages=false,colorlinks,hyperindex,bookmarksopen,linkcolor=red,citecolor=blue,urlcolor=blue]{hyperref}

\usepackage{mathrsfs}

\usepackage{epstopdf}

\usepackage{braket}

\bibpunct{[}{]}{;}{n}{,}{,}
\usepackage[hmargin=3cm,vmargin={3.5cm,4cm}]{geometry}

\newtheorem{prop}{Theorem}[section]
\newtheorem{thm}{Theorem}[section]

\newtheorem{rem}{Remark}[section]

\newtheorem{definition}{Definition}[section]

\newtheorem{coro}{Corollary}[section]

\usepackage{color}

\newcommand{\be}{\begin{eqnarray}}
\newcommand{\ee}{\end{eqnarray}}

\newcommand{\R}{\mathbb{R}}  
\newcommand{\C}{\mathbb{C}} 
\def\eg{{\it e.g. }} 
\def\ie{{\it i.e. }}

\newcommand{\ceil}[1]{\lceil #1 \rceil}

\def\eg{{\it e.g.}\ }
\def\ie{{\it i.e.}\ }

\numberwithin{equation}{section}

\allowdisplaybreaks

\begin{document}
\title{The fractional Dodson diffusion equation:\\
 a new approach}
	
		\author{Roberto Garra$^1$}
		\address{${}^1$ Department of Statistical Sciences, 
		 Sapienza, University of Rome. P.le Aldo Moro, 5,
		 Rome, Italy} 	
		\email{roberto.garra@sbai.uniroma1.it}
		
	    \author{Andrea Giusti$^2$}
		\address{${}^2$ Department of Physics $\&$ Astronomy, University of 	
    	    Bologna and INFN. Via Irnerio 46, Bologna, ITALY and 
	    	 Arnold Sommerfeld Center, Ludwig-Maximilians-Universit\"at, 
	    	 Theresienstra{\ss}e~37, 80333 M\"unchen, GERMANY.}	
 		\email{agiusti@bo.infn.it}
 		
 		\author{Francesco Mainardi$^3$}
    	    \address{${}^3$ Department of Physics $\&$ Astronomy, University of 	
    	    Bologna and INFN. Via Irnerio 46, Bologna, ITALY.}
			\email{francesco.mainardi@bo.infn.it}
 
    \keywords{Fractional Dodson equation, Fractional derivative of a function with respect to another function, Mittag-Leffler functions, Non-linear fractional diffusion.}

	\thanks{\textbf{Dedicated to Professor Tommaso Ruggeri on the occasion of his 70th anniversary.} \\
	To appear in \textbf{Ricerche di Matematica (2018)}, First online 17 January 2018, \\
	\textbf{DOI:} \href{https://doi.org/10.1007/s11587-018-0354-3}{10.1007/s11587-018-0354-3}.}	
	
    \date  {\today}

\begin{abstract}
   In this paper, after a brief review of the general theory concerning regularized derivatives and integrals of a function with respect to another function, we provide a peculiar fractional generalization of the $(1+1)$-dimensional Dodson's diffusion equation. For the latter we then compute the fundamental solution, which turns out to be expressed in terms of an M-Wright function of two variables. Then, we conclude the paper providing a few interesting results for some nonlinear fractional Dodson-like equations.     
       \end{abstract}
       
 \maketitle       
    
\section{Introduction}
     
     The Dodson diffusion equation arises in the context of cooling processes in geology \cite{dodson} and 
     it takes the form (see \eg \cite{crank}, pag.104-105)
     \begin{equation}\label{0}
     \frac{\partial c}{\partial t} = D_0 \, \exp(-\beta t) \, \frac{\partial^2 c}{\partial x^2}, \quad \beta = 1/\tau,
     \end{equation}
     where $\tau$ is the so called \textit{relaxation time}.
     
     Many experimental evidences show that, typically, in solid the diffusion coefficient depends on the temperature 
     according to the Arrhenius law. Nontheless, Dodson \cite{dodson} showed that, over a certain range, the temperature varies linearly with time leading to an approximate time-dependence of the diffusivity coefficient appearing in equation \eqref{0}. This represents a peculiar case of diffusion equation with time-dependent diffusivity.
     
     In the recent papers \cite{hristov1, hristov2}, Hristov studied the time-fractional generalization of the Dodson equation, obtained by replacing the first time-derivative appearing in \eqref{0} with a fractional derivative in the sense of Caputo, Riemann-Liouville and Caputo-Fabrizio.
      In particular, in \cite{hristov1} a formal fractionalization was performed whereas, in \cite{hristov2}, a generalized version of the Dodson equation was derived starting from first principles, i.e. the mass balance equation, by means of the exponential memory approach (i.e. Caputo-Fabrizio kernel). Furthermore, in the latter it is shown that, for large relaxation times, the considered model reduces to the ordinary Dodson equation.  
     
     One of the aim of these papers is, indeed, to compare the different approaches to fractional calculus and to obtain approximate solutions for the various fractionalizations of the Dodson equation. This problem is of particular interest since very few works on fractional diffusion equations with time-dependent diffusivity are present in the literature. Moreover, in this case, it is not simple to infer some exact solutions by means of standard methods like integral transforms. 
     
     From a physical point of view, the motivation for a fractional generalization of the Dodson equation is based on heuristic and explorative arguments that aim to understand the role of memory effects in more complex diffusion problems. In this perspective, the study of solid materials is of particular interest due to the success of fractional models in different fields of the applied sciences, see \eg \cite{CGM, noi, AG-FCAA, MMS}.
     Moreover, according to the present literature, memory effects play a relevant role in the physics of rocks (see for example \cite{geo} and references therein).
     
     In this paper, we adopt an alternative approach to the time-fractional Dodson equation. Specifically, we consider the following generalization of Eq.~\eqref{0}
     \begin{equation}\label{2}
     {}^C \left(e^{\beta t}\frac{\partial}{\partial t}\right)^\nu c(x,t) = D_0 \frac{\partial^2 c}{\partial x^2},
     \end{equation}
     involving a particular form of regularized Caputo-type fractional derivative of a function with respect to another function \cite{almeida}. The fractional operator appearing in \eqref{2} can be seen as a fractional power of the operator appearing in \eqref{0}, as we will explain in Section 2.
     We show the advantage of this approach in order to find the explicit form of the fundamenal solution and, therefore, to fully understand the role played by the memory with respect to the original model.
     By studying the generalized Dodson equation, we underline the relevance
     of fractional derivatives of a function with respect to another function to deal with mathematical models with time-dependent diffusivity coefficients. This mathematical tool can be useful also in relation to the studies of anomalous relaxation processes with time-dependent rate. For example, it can be used as an alternative approach to deal with the time-fractional model studied by de Oliveira et al. in \cite{meccanica}. Therefore, one of the aims of our paper is to underline the potential usefulness of fractional derivatives of a function with respect to another function to deal with realistic problems featured by time-varying coefficients and memory. We also remark that, lately, way too many definitions of fractional derivatives have been appearing in the literature, and that many of these are just specific realizations of more general (and physically relevant) fractional operators, see \eg the discussion in \cite{CG}.
     
    The paper is therefore organized as follows: first, we will find the similarity solution of the generalized Dodson equation \eqref{2} in terms of Wright functions, for which we will also provide some illuminating plots. Then, we will prove an interesting relation between the solution of the fractional Dodson equation and higher order diffusion equations with variable coefficients.
     
     \section{Preliminaries about Caputo fractional derivative of a function with respect to another function}
     
     Fractional derivatives of a function with respect to another function have been considered in the classical monograph by Kilbas et al. \cite{kilbas} (Section 2.5) and recently studied in detail by Almeida in \cite{almeida}.
     Here we briefly recall the main definitions and properties, refering to \cite{almeida} for a more complete discussion. We also observe that this approach has found some interesting applications in
     \cite{alm1}, \cite{noi} and \cite{orega} .
     
     \begin{definition}
      Let $\nu>0$, $I = (a,b)$ is the interval $-\infty \leq a < b \leq +\infty$, $f(t)\in C^1(I)$ an increasing function such that $f'(t)\neq 0$ for all $t\in I$, the fractional integral of a function g(t) 
          with respect to another function $f(t)$ is given by 
          \begin{equation}
          I^{\nu,f}_{a^+}g(t):=\frac{1}{\Gamma(\nu)}\int_a^t f'(\tau)
          (f(t)-f(\tau))^{\nu-1}g(\tau)d\tau.
          \end{equation}
     \end{definition}
    
     Observe that for $f(t) = t^{\alpha/\beta}$ we recover the definition of Erd\'elyi-Kober fractional integral that has recently foud many applications in various branches of physics and mathematics, see \eg \cite{gia1, FCAA2014}. In the case of $f(t)= \ln t$
     we have the Hadamard fractional integral and for $f(t)= t$, the Riemann-Liouville fractional integral (see \cite{kilbas}).
     
     The corresponding Caputo-type evolution operator is given in the following definition
     \begin{definition} \label{def-der}
     Let $\nu>0$, $n\in \mathbb{N}$,  $I = (a,b)$ is the interval $-\infty \leq a < b \leq +\infty$, $f(t), g(t)\in C^n(I)$ are two functions such that $f'(t)\neq 0 $ for all $t\in I$, the Caputo derivative of the function $g(t)$ with respect to the function $f(t)$ is given by
      \begin{equation}\label{2.1}
      {}^C \left(\frac{1}{f'(t)}\frac{d}{dt}\right)^\nu g(t) :=I_{a^+}^{n-\nu,f}\left(\frac{1}{f'(t)}\frac{d}{dt}\right)^n g(t),
          \end{equation}
          where $n = \ceil{\nu}$. 
     \end{definition}
    Notice that, we used a quite different notation with respect to the one adopted by Almeida in \cite{almeida}. This has been done in order to underline that one can, somehow, regard these operators as the fractional power counterpart of the operators with time-varying coefficients, \ie $\frac{1}{f'(t)}\left(\frac{\partial}{\partial t}\right)$.
     
     A relevant property of the operator (\ref{2.1}) is that  if $g(t) =(f(t)-f(a))^{\beta-1}$ with $\beta>1$, then (see Lemma 1 of \cite{almeida})
     \begin{equation}\label{2.2}
   {}^C \left(\frac{1}{f'(t)}\frac{d}{dt}\right)^\nu  g(t) = \frac{\Gamma(\beta)}{\Gamma(\beta-\nu)}(f(t)-f(a))^{\beta-\nu-1}.
     \end{equation}
     As a consequence, the composition of the Mittag-Leffler function with the function $f(t)$, that is
     \begin{equation} \label{2.4}
     g(t)= E_\nu(\lambda(f(t)-f(a))^\nu) = \sum_{k=0}^{\infty}\frac{\lambda^k((f(t)-f(a))^{\nu k}}{\Gamma(\nu k+1)},
     \end{equation}
     satisfies the equation
     \begin{equation}\label{eig}
     \displaystyle {}^C \left(\frac{1}{f'(t)}\frac{d}{dt}\right)^\nu g(t)= \lambda g(t),
     \end{equation}
     i.e. it is an eigenfunction of the operator $\displaystyle {}^C \left(\frac{1}{f'(t)}\frac{d}{dt}\right)^\nu$.\\
     The Mittag-Leffler functions play a central role in the studies about fractional differential equations and their applications, we refer for example to the recent monograph \cite{main} for a complete review on this topic.
     Compositions of Mittag-Leffler functions with other functions can be potentially useful, as recently discussed by Almeida in \cite{alm1}.
     The connection with integro-differential equations can give rigorous theoretical motivations for these studies.
     
     The operator appearing in the Dodson-type equation \eqref{2} corresponds to the particular choice of $f(t)$ in the general definition \eqref{2.2} such that $f'(t) =e^{-\beta t}$.  Hereafter we will take $a = 0 $ in the general definition.
     
     \section{The fractional Dodson equation}
     
     According to the definitions presented in the previous section, 
     we can now give the explicit expression for the operator appearing in our new generalization of the Dodson equation, that is
      \begin{equation}\label{3.1}
           {}^C \left(e^{\beta t} \, \frac{\partial}{\partial t}\right)^\nu  c(x,t) =\frac{1}{\Gamma(1-\nu)}\int_0^t \left(\frac{e^{-\beta \tau} - e^{-\beta t}}{\beta} \right) ^{- \nu} \, \frac{\partial c}{\partial \tau} \, d \tau,
    \end{equation}
 	for $\nu \in (0,1)$.
 	
 	Indeed, if $0 < \nu < 1$ and $a = 0$ then from Definition \ref{def-der} we immediately infer that
 	\begin{equation}
 	{}^C \left(\frac{1}{f'(t)} \, \frac{\partial}{\partial t} \right)^\nu c(x,t) = \frac{1}{\Gamma (1 - \nu)} \int _0 ^t \left( f(t) - f(\tau) \right)^{- \nu} \, 
 	\frac{\partial c}{\partial \tau} \, d \tau \, .
 	\end{equation}
 	Hence, if we now consider the operator in \eqref{3.1}, it is easy to see that $f' (t) = e^{- \beta t}$, from which it follows that
 	$$f(t) = - \frac{1}{\beta} \, e^{- \beta t} + K \, , $$
 	where $K \in \R$ is just an integration constant.

 	\begin{thm}
 	The fundamental solution of the generalized Dodson equation  
 	   \begin{equation}\label{3.2}
 	     {}^C \left(e^{\beta t}\frac{\partial}{\partial t}\right)^\nu c(x,t) =  \frac{\partial^2 c}{\partial x^2}, \quad \nu\in(0,1), \ x\in \mathbb{R}, t>0,
 	     \end{equation}
 	is given by
 	\begin{equation} \label{sol1}
 	c(x,t) = \frac{1}{2} \, \mathbb{M} _{\nu / 2} \left[ |x| , \, \left(\frac{1 - e^{-\beta t}}{\beta} \right) \right] \, ,
 	\end{equation}
 	where
 	\begin{equation}
 	\begin{split}
 	&M_{\alpha}(z) = \sum_{n=0}^\infty\frac{(-z)^n}{n!\Gamma(-\nu n+1-\nu)} \, , \qquad z \in \C \, , \,\, 0 < \alpha < 1, \\
 	&\mathbb{M}_{\alpha} (y , \, \lambda) = \frac{1}{\lambda ^\alpha} \, M _{\alpha} \left( \frac{y}{\lambda ^\alpha}\right) \, , \qquad y, \lambda \in \R ^+ \, ,
 	\end{split}
 	\end{equation}
 are the \textit{M}-Wright function\footnote{For further details on the M-Wright function, and its applications, we refer the interested reader to a very interesting surveys \cite{MMP, Gianni}. From an historical perspective, it is also worth remarking that \textit{M}-Wright function was first introduced in the occasion of the \textit{7th WASCOM} conference, see \cite{WASCOM}.} (also known as Mainardi function) and the \textit{M}-Wright function of two variables, respectively.
 	\end{thm}
 	
 	\begin{proof}
 	By taking the Fourier transform of \eqref{3.2}, we obtain 
 	   \begin{equation}\label{k}
 	     {}^C \left(e^{\beta t}\frac{\partial}{\partial t}\right)^\nu \widehat{c}(k,t) = -k^2 \widehat{c}(k,t),
 	     \end{equation}
    whose solution is given by 
    \begin{equation}\label{k1}
    \widehat{c}(k,t) = E_{\nu} \left[-k^2 \left(\frac{1-e^{-\beta t}}{\beta}\right)^\nu \right] \, ,
    \end{equation}
    which is a direct consequence of the discussion concerning the eigenvalue problem for the operator \eqref{2.1} presented in the previous section.
	
	Now, recalling that (see \cite{MMP})
	\begin{equation}
	\mathcal{F} \left\{ \frac{1}{2} \, \mathbb{M} _{\alpha/2} (|x|, \lambda) \, ; \, x \mapsto k \right\} = E _{\alpha} \left( - k^2 \, \lambda ^\alpha \right) \, ,
	\end{equation}	    
   where $\mathcal{F} [f(x) \, ; \, x \mapsto k ] = \widehat{f} (k)$, then it is easy to see that the inverse Fourier transform of \eqref{k1} provides exactly the claimed result. 
    \end{proof}      
 
For sake of clarity, in Figs.~\ref{Fig-1} and~\ref{Fig-2} we show some fixed-time plots of Eq.~\eqref{sol1}.
 
\newpage 
 
\begin{figure}[h!]
	\centering
	\includegraphics[scale=0.55]{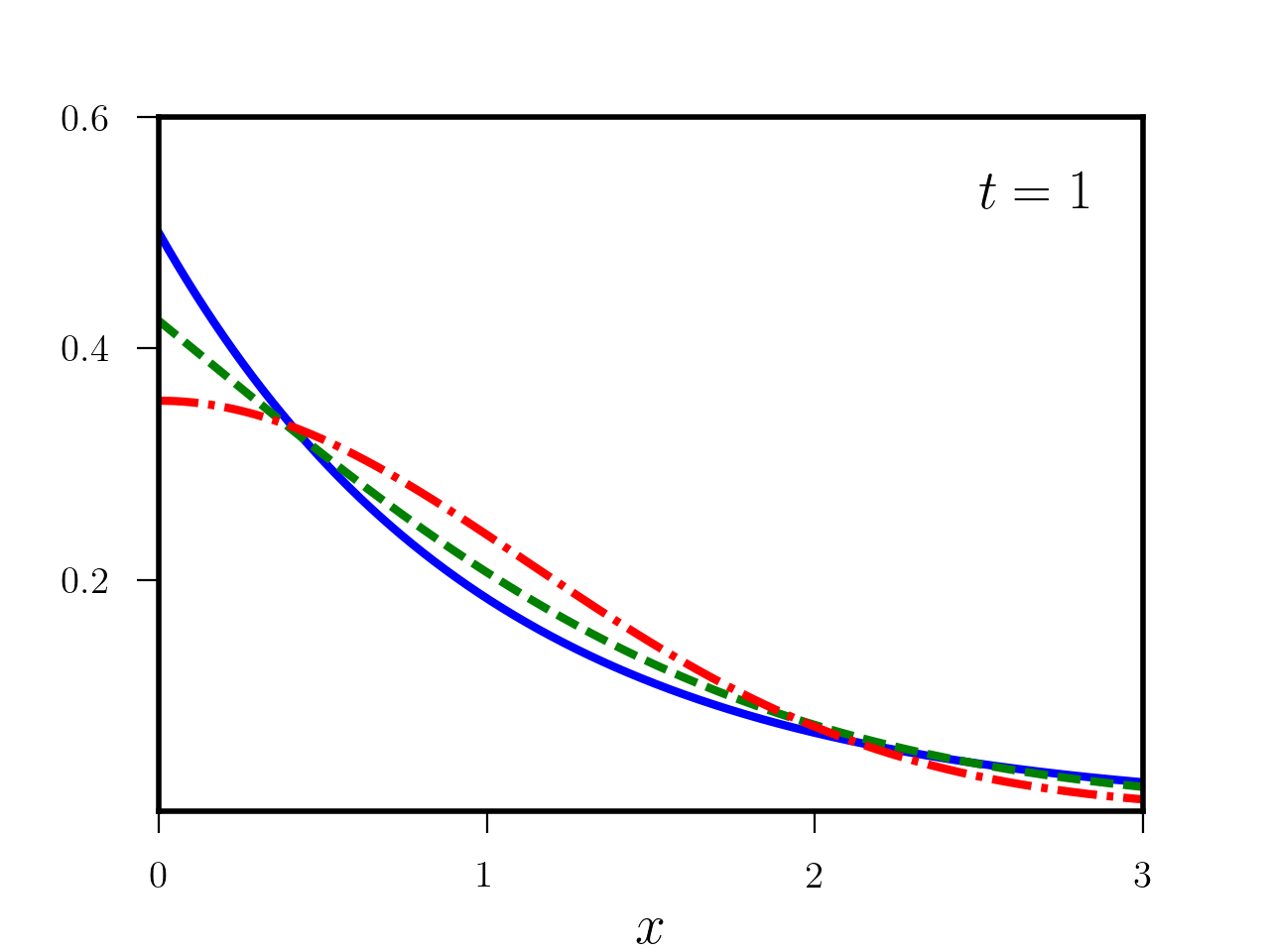}
	\caption{Plot of the solution \eqref{sol1} as a function of $x$ at a fixed time $t = 1$ (setting $\beta = 1$). The solid line corresponds to $\nu = 0.001$, the dashed line is obtained by setting $\nu = 0.7$ whereas the dash-point line is given by $\nu = 1$.} \label{Fig-1}
	\end{figure}
	
	\begin{figure}[h!]
	\centering
	\includegraphics[scale=0.55]{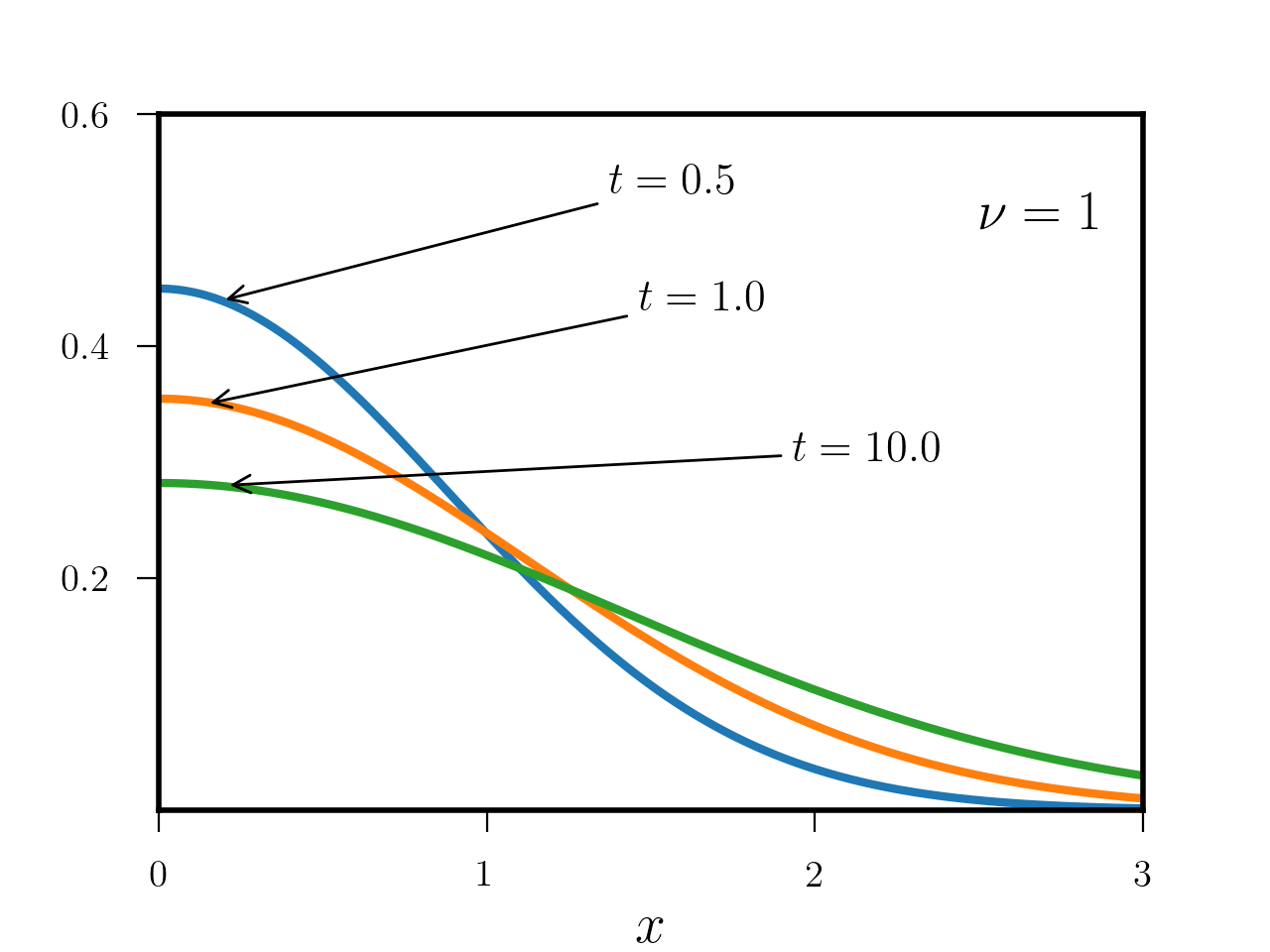}
	\caption{Plot of \eqref{gaussian} as a function of $x$ at different times $t = 0.5, 1.0, 10.0$ ($\beta = 1$).} \label{Fig-2}
	\end{figure}

	\begin{figure}[h!]
	\centering
	\includegraphics[scale=0.6]{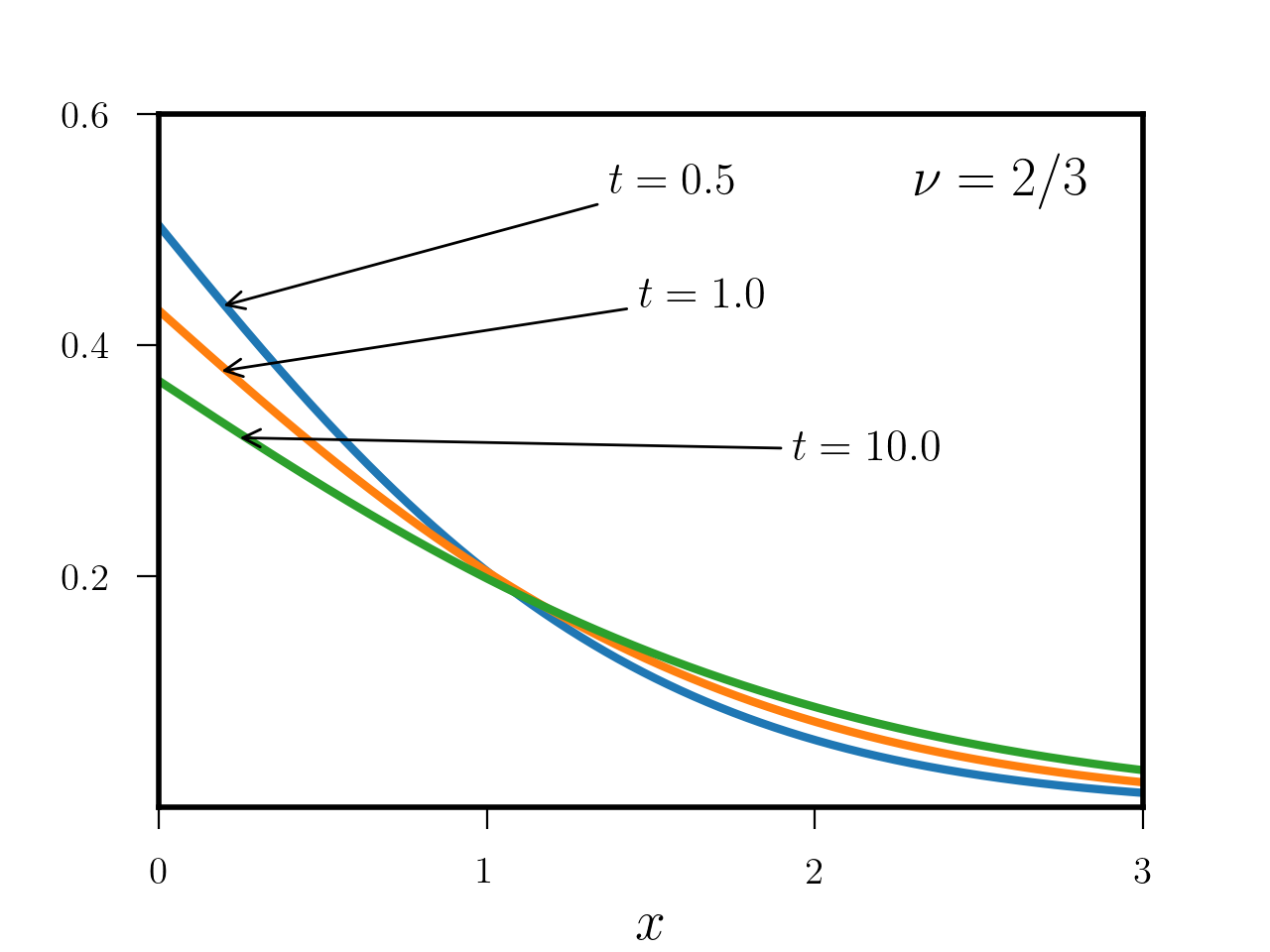}
	\caption{Plot of \eqref{Ai} as a function of $x$ at different times $t = 0.5, 1.0, 10.0$ ($\beta = 1$).} \label{Fig-3}
	\end{figure}

 	\begin{rem}
 	Observe that for $\nu = 1 $, we recover the exact similarity solution of the Dodson equation (see \cite{hristov1}, Section 3.2), that is given by 
 	\begin{equation} \label{gaussian}
 	c(x,t) =\frac{1}{\sqrt{4 \,\pi} \, \sqrt{1-e^{- t}}}  \, \exp\left(-\frac{x^2}{4(1-e^{- t})}\right),
 	\end{equation}
 	where we take $\beta = D_0 = 1$ for simplicity. By using our approach, we obtain, therefore, the composition of the classical solution of the time-fractional diffusion equation with a new time variable $t' = 1-e^{-\beta t}$. This is not surprising, since, also in the non-fractional case, we can reduce the Dodson equation to the classical heat equation by using this change of variable.
 	On the other hand the relevance of our approach stands in the facilitiy to obtain a simple solution to a fractional problem with time-varying diffusivity.  
 	\end{rem}      
     
     \begin{thm}
     The solution of the fractional Dodson equation \eqref{3.2} coincides with the solution to the higher order diffusion equation with time-dependent diffusivity, \ie
     \begin{equation}
     \frac{\partial c}{\partial t}  = (-1)^k e^{-\beta t} \frac{\partial^n c}{\partial x^n} \, ,
     \end{equation}
     when $\nu = 2/n$, for any $n\in \mathbb{N}$ such that $n>2$.
     \end{thm}
   The proof results by plain calculations.
   
   \begin{coro}
    The fundamental solution of the fractional Dodson equation \eqref{3.2} for $\nu = 2/3$ and $\beta = 1$ (see Fig.~\ref{Fig-3}) is given by 
    \begin{equation} \label{Ai}
    c(x,t) = \frac{3^{2/3}}{2 \cdot\sqrt[3]{(1-e^{- t})}}Ai\left(\frac{|x|}{\sqrt[3]{3(1-e^{- t})}}\right) \, ,
    \end{equation}
    that corresponds to the solution of the linearized KdV equation with time-varying coefficient
    \begin{equation}
    \frac{\partial c}{\partial t} = -e^{-t}\frac{\partial^3 c}{\partial x^3}.
    \end{equation}
   \end{coro}  
    The last Corollary is a consequence of Theorem 3.2 and of the well-known fact that (see for example \cite{gianni})
    \begin{equation}
    M_{1/3}(z) = 3^{2/3}Ai(z/3^{1/3}).
    \end{equation}
    We observe, for the sake of completness, that dispersion relation for time-fractional linearized KdV equations have been recently studied by Colombaro et al. in \cite{col}. 

    \section{Nonlinear fractional Dodson-type equations: some explicit results}
    
    In this section we provide some results for the following nonlinear counterpart of the fractional Dodson equation
    \begin{equation}\label{nld}
    {}^C \left(e^{\beta t}\frac{\partial}{\partial t}\right)^\nu c= \frac{\partial^2 c^m}{\partial x^2}, \quad m>0.
    \end{equation}
    This can be viewed as a time-fractional nonlinear diffusion equation with a time-dependent diffusivity. Observe that, when $\beta = 0$ it is essentially a time-fractional porous medium equation, recenty studied in the mathematical literature by Dipierro et al. in \cite{dipierro}.
    
    In this case we have the following result
    \begin{prop}
    Let $\nu \in (0,1)$, $m\in \mathbb{R}^+\setminus \{1\}$, then equation \eqref{nld} admits a solution of the form
    \begin{equation}
    c(x,t) = \bigg[\frac{(m-1)^2}{2m(m+1)}\frac{\Gamma(1-\frac{\nu}{m-1})}{\Gamma(1-\frac{\nu}{m-1}-\nu)}\bigg]^{\frac{1}{m-1}}\left[\frac{\beta^\nu \ x^2}{(1-e^{-\beta t})^\nu}\right]^{\frac{1}{m-1}}, \quad t >0.
    \end{equation}
    \end{prop}
    \begin{proof}
    We consider the following \textit{ansatz} on the form of the similarity solution of \eqref{nld}
    \begin{equation}
  	c(x,t) = f(t)x^{\frac{2}{m-1}},  
    \end{equation}
    based on the invariance property
    \begin{equation}
    \frac{\partial^2 }{\partial x^2} \left(x^{\frac{2}{m-1}}\right)^m= \left(\frac{2m}{m-1}\right)\left(\frac{2m}{m-1}-1\right)x^{\frac{2}{m-1}} \, .
    \end{equation}
    By substitution it follows that $f(t)$ has to satisfy the following equation
    \begin{equation}\label{nld1}
    {}^C \left(e^{\beta t}\frac{\partial}{\partial t}\right)^\nu f(t)= 
    \left(\frac{2m}{m-1}\right)\left(\frac{2m}{m-1}-1\right)f^m(t) \, .
      \end{equation}
    We now assume that \eqref{nld1} admits a solution of the form
    \begin{equation}\label{cia}
    f(t) = C_1\left(\frac{1-e^{-\beta t}}{\beta}\right)^\gamma,
    \end{equation}
    where $\gamma$ and $C_1$ are two, yet unknown, functions of $m$ and $\nu$. 
    
    By substituting \eqref{cia} in \eqref{nld1} and taking profit of the property
    \eqref{2.2} of the fractional operator appearing in \eqref{nld}, we infer
    that
    \begin{align}
     \nonumber {}^C \left(e^{\beta t}\frac{\partial}{\partial t}\right)^\nu
    C_1\left(\frac{1-e^{-\beta t}}{\beta}\right)^\gamma &=
    C_1\frac{\Gamma(\gamma+1)}{\Gamma(\gamma+1-\nu)}\left(\frac{1-e^{-\beta
    t}}{\beta}\right)^{\gamma-\nu}\\
     &\nonumber =
    \left(\frac{2m}{m-1}\right)\left(\frac{2m}{m-1}-1\right)C_1^m
    \left(\frac{1-e^{-\beta t}}{\beta}\right)^{\gamma m}
    \end{align}
    and this equality is clearly satisfied only if
    \begin{align}
    \nonumber &\gamma= -\frac{\nu}{m-1},\\
    \nonumber & C_1  =
    \Bigg[\frac{1}{\left(\frac{2m}{m-1}\right)\left(\frac{2m}{m-1}-1\right)}\frac{\Gamma(1-\frac{\nu}{m-1})}{\Gamma(1-\frac{\nu}{m-1}-\nu)}\Bigg]^{\frac{1}{m-1}}
    \end{align}
    as claimed.  
    \end{proof}
    
    \begin{rem}
    Observe that the obtained solution can be considered as a \textit{fractional} counterpart of the similarity solution 
    $$ c(x,t) \sim \left(\frac{\beta x^2}{1-e^{-\beta t}}\right)^{\frac{1}{m-1}} $$
    of the nonlinear diffusive Dodson-type equation
    \begin{equation}
    \frac{\partial c}{\partial t}= e^{-\beta t}\frac{\partial^2 c^m}{\partial x^2}.
    \end{equation}
    \end{rem}
    
    \begin{rem}
     The sign of the solution obtained above clearly depends on $m$ and $\nu$. 
     Moreover, this similarity-type solution is peculiar of non-linear Dodson-like equations since, clearly, $m$ and $\beta$ have to be non-vanishing.
    \end{rem}
    
    We now consider the nonlinear diffusive equation \eqref{nld1} with a linear absorption term, for which we obtain the following
    \begin{prop}
    The equation 
    \begin{equation}\label{nld3}
        {}^C \left(e^{\beta t}\frac{\partial}{\partial t}\right)^\nu c= \frac{\partial^2 c^m}{\partial x^2}-c, \quad m>0,
        \end{equation}
     admits a solution of the form 
     \begin{equation}
     c (x, t) = E_{\nu}\bigg[-\left(\frac{1-e^{-\beta t}}{\beta}\right)^\nu\bigg] \, x^{1/m}
     \end{equation}
    \end{prop}
     \begin{proof}
     Let us assume that the equation admits a solution, via separation of variables, of the form $x^{1/m}f(t)$, then if we plug in this \textit{ansaz} we get
     \begin{equation}
     ^C \left(e^{\beta t}\frac{\partial}{\partial t}\right)^\nu  x^{1/m} f(t)= -x^{1/m}f(t),
     \end{equation}
     whose solution is given by (see \eqref{eig})
     \begin{equation}
     \nonumber f(t) =E_{\nu}\bigg[-\left(\frac{1-e^{-\beta t}}{\beta}\right)^\nu\bigg],
     \end{equation}
     as claimed.
     \end{proof}
     It is worth noting that this solution corresponds to a particular choice of the initial condition $c(x,0)= x^{1/m}$.
     Other exact results can be simply obtained by means of generalized methods of separation of variables. 
     
     	\section*{Acknowledgments}
		The work of the authors has been carried out in the framework of the activities of the National Group of Mathematical Physics (GNFM, INdAM).
	
	Moreover, the work of A.G. has been partially supported by \textit{GNFM/INdAM Young Researchers Project} 2017 ``Analysis of Complex Biological Systems''.

\end{document}